\documentclass[10pt, conference]{IEEEtran}
\usepackage{amsmath}
\usepackage{amssymb}
\usepackage{amsfonts}
\usepackage{amscd}
\usepackage{mathrsfs}
\usepackage{mdframed}
\usepackage{graphicx}
\usepackage{color}
\usepackage{url}
\usepackage{bm}
\usepackage{algorithm}
\usepackage{algorithmic}

\newtheorem{theorem}{Theorem}
\newtheorem{lemma}{Lemma}
\newtheorem{definition}{Definition}

\begin{document}

%%\title{On the Robustness of Algebraic STBCs to Coefficient Quantization}
%\title{Full-Rate, Full-Diversity STBCs Robust to Coefficient Quantization in $2 \times 2$ MIMO Systems}
%
%\author{\authorblockN{J.~Harshan,~\IEEEmembership{Member,~IEEE}, A.~Sakzad,~\IEEEmembership{Member,~IEEE}, and E.~Viterbo,~\IEEEmembership{Fellow,~IEEE}}\thanks{The authors are with the Department of Electrical and Computer Systems Engineering,
%Monash University, Melbourne, Victoria, Australia. This work was performed at the Monash Software Defined Telecommunications (SDT) Lab and was supported by the Monash Professional Fellowship and the Australian Research Council under Discovery grants ARC DP 130100103. Emails: $\{$harshan.jagadeesh, amin.sakzad, emanuele.viterbo$\}$@monash.edu}}
%\author{\authorblockN{J.~Harshan, Amin Sakzad, and Emanuele Viterbo}\thanks{This work was performed at the Monash Software Defined Telecommunications Lab
%supported by the Talented Enhancement Scheme through the Monash Professorial Fellowship (MPF) program.} %Emails: $\{$amin.sakzad, harshan.jagadeesh, emanuele.viterbo$\}$@monash.edu}
%\authorblockA{Department of Electrical and Computer Systems Engineering,\\
%Monash University, Melbourne, Victoria, Australia}
%$\{$harshan.jagadeesh, amin.sakzad, and emanuele.viterbo$\}$@monash.edu}

%\title{Storage Efficient Erasure Coded Archival of Versioned Data}
\title{Integer-Forcing Linear Receivers: A Design Criterion for Full-Diversity STBCs}
\author{J. Harshan$^{\dagger}$, Amin Sakzad$^{\star}$, Emanuele Viterbo$^{\ast}$\\
$^{\dagger}$Advanced Digital Sciences Center, Singapore, \\$^{\star}$Clayton School of Information Technology, Monash University, Australia,\\ $^{\ast}$Department of Electrical and Computer Systems, Monash University, Australia\\
Email: harshan.j@adsc.com.sg, amin.sakzad@monash.edu, emanuele.viterbo@monash.edu\\}

\maketitle

\begin{abstract}
In multiple-input multiple-output (MIMO) fading channels, the design criterion for full-diversity space-time block codes (STBCs) is primarily determined by the decoding method at the receiver. Although constructions of STBCs have predominantly matched the maximum-likelihood (ML) decoder, design criteria and constructions of full-diversity STBCs have also been reported for low-complexity linear receivers. A new receiver architecture called Integer-Forcing (IF) linear receiver has been proposed to MIMO channels by Zhan \emph{et al.} which showed promising results for the high-rate V-BLAST encoding scheme. In this work we address the design of full-diversity STBCs for IF linear receivers. We derive an upper bound on the probability of decoding error, and show that STBCs that satisfy the non-vanishing singular value (NVS) property provide full-diversity for the IF receiver. We also present simulation results to demonstrate that linear designs with NVS property provide full diversity for IF receiver. As a special case of our analysis on STBCs, we present an upper bound on the error probability for the V-BLAST architecture presented by Zhan \emph{et al.}, and demonstrate that the IF linear receivers provide full receive diversity. Our results supplement the existing outage probability based results for the IF receiver. 

\end{abstract}

%%%%%%%%%%%%%%%%%%%%%%%%%%%%%%%%%%%%%%%%%%%%%%%%%%%%%%%%
\section{Introduction and Preliminaries}
\label{sec1}

%Talk about space time coding for MIMO with ML and other receivers

Space-time coding is a powerful transmitter-side technique that assists reliable communication over multiple-input multiple-output (MIMO) fading channels. For a MIMO channel with $n_{t}$ transmit and $n_{r}$ receive antennas, a space-time block code (STBC) denoted by $\mathcal{C} \subset \mathbb{C}^{n_{t} \times T}$ is a finite set of complex matrices used to convey $\mbox{log}_{2}(|\mathcal{C}|)$ information bits to the destination \cite{TSC}. To recover the information bits all the observations collected across $n_{r}$ receive antennas and $T$ time slots at the destination are  appropriately processed by a suitable decoder $\mathcal{D}$, e.g., maximum-likelihood (ML) decoder, zero-forcing (ZF) receiver, or a minimum mean square error (MMSE) receiver.
%\begin{definition}
%An STBC $\mathcal{C}$ is said to provide full-diversity for decoder $\mathcal{D}$ if at high signal-to-noise ratio (SNR), the average probability of decoder error behaves as \cite{TSC}
%\begin{equation*}
%\mbox{Pr}(\hat{\textbf{X}} \neq \textbf{X}) \leq \frac{c}{\mbox{SNR}^{n_{t}n_{r}}},
%\end{equation*}
%where $\hat{\textbf{X}}$ is the decoded codeword, $\textbf{X} \in \mathcal{C}$ is the transmitted codeword, and $c$ is some constant independent of $\mbox{SNR}$.
%\end{definition}
%
%A well-known method to generate an $n_{t} \times T$ STBC is to allow the variables of an $n_{t} \times T$ linear design $\textbf{X}_{\mathcal{LD}}(x_{1}, x_{2}, \ldots, x_{K})$ take values from a finite set of complex numbers. For such STBCs, the symbol-rate is defined as $\mathcal{R} = \frac{K}{T}$ complex symbols per channel use. An STBC which is designed based on the well-known rank criterion \cite{TSC} is known to provide full-diversity for a MIMO channel. However, such a criterion applies only when the STBCs are decoded using the optimal ML decoder. 

\indent Considering the high computational complexity of the ML decoder, many research groups have addressed the design and construction of full-diversity STBCs that are matched to suboptimal \emph{linear receivers} such as the ZF and the MMSE receivers \cite{ZLW1}-\cite{WXYL}. These linear receivers reduce the complexity of the decoding process by trading off some error performance with respect to ML decoder. In \cite{ShX} a new design criterion for full-diversity STBCs matched to ZF receivers is proposed, which imposes a constraint on the symbol-rate of STBCs. In particular, it has been proved that the symbol-rate of such STBCs is upper bounded by one. For some code constructions matched to ZF and MMSE linear receivers, we refer the reader to \cite{ZLW1}, \cite{WXYL}. In summary, a rate loss is associated with the design of full-diversity STBCs compliant to ZF and MMSE receivers \cite{Aria13}. A comparison of the decoding complexity, diversity, and the symbol-rate of STBCs for these decoders is summarized in Table \ref{complexity_table} (given in next page). 
%Other than the ZF and the MMSE receivers, STBCs have also been designed for other suboptimal decoders in MIMO channels \cite{GuX}-\cite{LPS3}. 
%
%Talk about IF linear receivers

\indent A new receiver architecture called integer forcing (IF) linear receiver has been recently proposed~\cite{zhan12} to attain higher rates with reduced decoding complexity. In such a framework, the source employs a layered transmission scheme and transmits independent codewords simultaneously across the layers. This has been referred to as the \emph{V-BLAST encoding scheme} in ~\cite{zhan12}. At the receiver side, each layer is allowed to decode an integer linear combination of transmitted codewords, and then recover the information by solving a system of linear equations.
%An outage probability based analysis has been presented to demonstrate that IF linear receivers deliver receive diversity of $n_{r}$.
%Talk about no rate penalty for IF linear receivers 

Although IF receivers are known to work well with the V-BLAST scheme \cite{WeC}-\cite{DoE}, not many works have investigated the suitability of IF receivers to decode STBCs. In \cite{OrE} IF receiver has been applied to decode a layered transmission scheme involving perfect STBCs. Such an architecture has been shown to achieve the capacity of any Gaussian MIMO channel upto a gap that depends on the number of transmit antennas. An interesting question that arises from \cite{OrE} is: What is the design criterion for full-diversity STBCs for IF linear receivers?

In order to answer the above question, it is paramount to characterize the structure of good STBCs for IF receivers. Towards that end, we study the error performance of IF receivers along the lines of \cite{TSC, ShX}, and propose a design criterion for constructing full-diversity STBCs. The contributions of this paper is as follows: We study the application of IF linear receivers to decode STBCs in MIMO channels. We present a decoder error analysis for the IF receiver in order to obtain a design criterion for full-diversity STBCs. We show that STBCs that satisfy a criterion called the non-vanishing singular value (NVS) criterion provide full-diversity for the IF linear receiver. This is a stronger condition than the rank criterion \cite{TSC} for the ML decoder. In \cite{zhan12} an outage probability based error analysis of the IF receiver is presented for the V-BLAST scheme. As a special case of our error analysis on STBCs, we derive an upper bound on the error probability for the architecture presented in \cite{zhan12} and show that the IF linear receivers provide full receive diversity for such a scheme. Thus, our results supplement the existing outage based results for the IF receiver.

The problem of constructing STBCs based on the minimum singular value criterion is not entirely new. In \cite[Ch. 9]{Tse_visw} it has been shown that maximising the minimum singular value of the difference of codeword matrices provides the \emph{approximate universality} property for STBCs in Multiple-Input Single-Output (MISO) channels. It is to be noted that the criterion in \cite{Tse_visw} applies for the ML decoder. However, in this work, we show that the NVS criterion is applicable for the IF receiver for MIMO channels (not only MISO channels).

%\indent In summary, among the class of linear receivers, we show that IF receivers can admit STBCs with larger symbol-rate than that of the MMSE and ZF receivers. As shown in Table. \ref{complexity_table}, our results highlight that unlike the traditional linear receivers such as ZF and MMSE receivers, the design criterion does not impose limitation on the symbol-rate of STBCs for IF receivers. 
%This is confirmed by the admittance of high-rate perfect STBCs by the IF receiver.
\begin{table}
\begin{center}
\caption{Comparison of STBCs for various receivers: $n_{t}$ and $n_{r}$ denote the number of transmit and receive antennas, respectively.}
\begin{tabular}{|c|c|c|c|} \hline
Approach & Decoding & Spatial & Symbol-\\
& Complexity & Diversity & Rate\\
\hline
%\multirow{5}{*}{Integer-Forcing Linear receiver}
{\textbf ML} & high & $n_{t}n_{r}$ & $\leq \mbox{min}(n_{t}, n_{r})$\\
{\textbf ZF \& MMSE} & low & $n_{t}n_{r}$ & $\leq 1$\\
%{\textbf ZF \& MMSE} & low & $n_{t} - n_{r} + 1$ & $\mbox{min}(n_{t}, n_{r})$\\
{\textbf IF} & low & $n_{t}n_{r}$ & $\leq \mbox{min}(n_{t}, n_{r})$\\
\hline
\end{tabular}
\end{center}
\label{complexity_table}
\end{table}

%\indent The rest of the paper is organized as follows: In Section \ref{sec2}, we introduce the system model of STBCs for MIMO channel, and list the decoding procedure for STBCs based on the IF receiver. In Section \ref{sec4}, we propose a probability of error analysis for the decoder and propose a design criterion for full-diversity STBCs. In Section \ref{sec5}, we use the decoder error analysis to prove the diversity results of the layered scheme in \cite{zhan12}, while in Section \ref{sec6}, we discuss the problem of constructing STBCs for IF linear receivers. Finally, we present the concluding remarks and directions for future work in Section \ref{sec7}.

\indent {\em Notations}. Boldface letters are used for vectors, and capital boldface letters for matrices. We let $\mathbb{R}$, $\mathbb{C}$, $\mathbb{Z}$, $\mathbb{Q}$, and $\mathbb{Z}[\imath]$ denote the set of real numbers, complex numbers, integers, rational numbers, and the Gaussian integers, respectively, where $\imath^2 = -1$. We let ${\textbf I}_n$ and ${\textbf 0}_n$ denote the $n\times n$ identity matrix and zero matrix and the operations $(\cdot)^T$ and $(\cdot)^H$ denote transposition and Hermitian transposition. We let $| \cdot |$ and $\| \cdot \|$ denote the absolute value of a complex number and the Euclidean norm of a vector, respectively. The operation $\mathbb{E}(\cdot)$ denotes mean of a random variable. We let $\lfloor x \rceil$ and $\lfloor {\textbf v} \rceil$ denote the closest integer to $x$ and the component-wise equivalent operation. The symbol $\textbf{X}_{j, m}$ denotes the element in the $j$-th row and $m$-th column of $\textbf{X}$. For a matrix $\textbf{X}$, the Frobenious norm $\sqrt{\sum_{j} \sum_{m}|\textbf{X}_{j, m}|^2}$ is denoted by $\|\textbf{X}\|_{F}$. The symbol $\mathcal{N}_{c}(0, 1)$ denotes circularly complex Gaussian distribution with mean zero and unit variance. For an $n_{t} \times T$ matrix $\textbf{X}$, the symbol $\sigma_{j}(\textbf{X})$ denotes the $j$-th singular value of $\textbf{X}$ for $1 \leq j \leq n_{t}$. The real and imaginary parts of a complex matrix $\textbf{X}$ is denoted by $\mbox{Re}(\textbf{X})$ and $\mbox{Im}(\textbf{X})$, respectively. The symbol $\mbox{Pr}(\cdot)$ denotes the probability operator.
%%%%%%%%%%%%%%%%%%%%%%%%%%%%%%%%%%%%%%%%%%%%%%%%%%%%%%%%%%%%%%%%%%%%%%%%%%%%%%%%%%%%%%
%New Section
%%%%%%%%%%%%%%%%%%%%%%%%%%%%%%%%%%%%%%%%%%%%%%%%%%%%%%%%%%%%%%%%%%%%%%%%%%%%%%%%%%%%%
\section{System Model}
\label{sec2}
The $n_{t} \times n_{r}$ MIMO channel consists of a source and a destination terminal equipped with $n_{t}$ and $n_{r}$ antennas, respectively. For $1 \leq i \leq n_{t}$ and $1 \leq j \leq n_{r}$, the channel between the $i$-th transmit antenna and the $j$-th receive antenna is assumed to be flat fading and denoted by the complex number $\textbf{H}_{i, j}$. Each $\textbf{H}_{i, j}$ remains constant for a block of $T$ ($T \geq n_{t}$) complex channel uses and is assumed to take an independent realization in the next block. Statistically, we assume $\textbf{H}_{i, j} \sim ~\mathcal{N}_{c}(0, 1) ~\forall i, j$ across quasi-static intervals. The source conveys information to the destination through an $n_{t} \times T$ STBC denoted by $\mathcal{C}$. We assume that a linear design
\begin{equation}
\label{linear_design}
\textbf{X}_{\mathcal{LD}}(s_{1}, \ldots, s_{2K}) = \sum_{k = 1}^{2K} \textbf{D}_{k} s_{k},
\end{equation}
in $2K$ real variables $\textbf{s} = [s_{1} ~s_{2} ~\ldots ~s_{2K}]^{T}$ is used to generate $\mathcal{C}$ by taking values from an underlying integer constellation $\mathcal{S} \subset \mathbb{Z}$. Here, the set $\{\textbf{D}_{k} \in \mathbb{C}^{n_{t} \times T}\}_{k = 1}^{2K}$ contains the weight matrices of the design. Since we use the IF linear receiver to decode the STBC, we assume that $\mathcal{S}$ is a finite ring $\mathbb{Z}_{\sqrt{M}} = \left\{0, 1, \ldots, \sqrt{M} - 1 \right\}$ for some $M$, an even power of $2$. The symbols of $\mathcal{S}$ are appropriately shifted around the origin to reduce the transmit power, and subsequently  reverted back at the receiver to retain the ring structure on $\mathcal{S}$. If $\textbf{X}(\textbf{s}) \in \mathcal{C}$ denotes a transmitted codeword matrix such that $\mathbb{E}[|\textbf{X}_{i, t}|^{2}]= 1 ~\mbox{ for } 1 \leq i \leq n_{t}, 1 \leq t \leq T$, then the received matrix $\textbf{Y} \in \mathbb{C}^{n_{r} \times n_{t}}$ at the destination is given by
\begin{equation}
\label{signal_model}
{\textbf Y} = \sqrt{\frac{P}{n_{t}}}{\textbf H}{\textbf X}(\textbf{s}) + {\textbf Z},
\end{equation}
where $\textbf{H} \in \mathbb{C}^{n_{r} \times n_{t}}$ denotes the channel matrix, $\textbf{Z} \in \mathbb{C}^{n_{r} \times T}$ denotes the AWGN with its entries that are i.i.d. as $\mathcal{N}_{c}(0, 1)$. With this, the average receive signal power-to-noise ratio (SNR) per receive antenna is $P$. Throughout the paper, we assume a coherent MIMO channel where only the receiver has the complete knowledge of $\textbf{H}$.
%%%%%%%%%%%%%%%%%%%%%%%%%%%%%%%%%%%%%%%%%%%%%%%%%%%%%%%%%%%%%%%%%%%%%%%%%%%%%%%%%%%%%%
\begin{figure}
\includegraphics[width=3.5in]{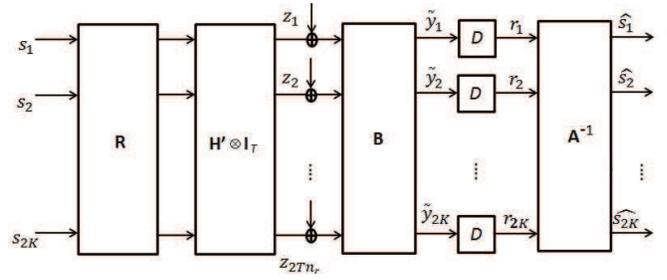}
\vspace{-1cm}
\caption{IF linear receiver to decode STBCs where $\textbf{R}$ is a code matrix obtained from vectorizing the components of the weight matrices $\{\textbf{D}_{k}\}_{k = 1}^{2K}$ in \eqref{linear_design}.}
\label{model1}
\end{figure}
%New Section

\indent In the next subsection, we discuss the decoding procedure for STBCs based on the IF receiver.

%%%%%%%%%%%%%%%%%%%%%%%%%%%%%%%%%%%%%%%%%%%%%%%%%%%%%%%%%%%%%%%%%%%%%%%%%%%%%%%%%%%%%
\subsection{IF decoder for STBCs}

Since $\mathcal{C}$ is a linear dispersion code, the received matrix $\textbf{Y}$ in \eqref{signal_model} can be vectorized to obtain a noisy linear model as
\begin{eqnarray}
\label{linear_model}
\textbf{y} = \sqrt{\frac{P}{n_{t}}}\mathbf{\mathcal{H}}\textbf{s} + \textbf{z},
\end{eqnarray}
where $\mathcal{H} \in \mathbb{R}^{2n_{r}T \times 2K}$ is given by 
\begin{equation}
\label{H_dash1}
\mathcal{H} = (\textbf{H}' \otimes \textbf{I}_{T})\textbf{R},
\end{equation}
such that
\begin{equation}
\label{H_dash}
\textbf{H}' = \left[\begin{array}{rr}
\mbox{Re}(\textbf{H}) & -\mbox{Im}(\textbf{H})\\
\mbox{Im}(\textbf{H})  & \mbox{Re}(\textbf{H})\\
\end{array}\right] \in \mathbb{R}^{2n_{r} \times 2n_{t}},
\end{equation}
and $\textbf{R} \in \mathbb{R}^{2n_{t}T \times 2K}$ is a code matrix obtained from vectorizing the components of the weight matrices $\{\textbf{D}_{k}\}_{k = 1}^{2K}$. Here, the symbol $\otimes$ denotes the Kronecker product operator. After suitable scaling, \eqref{linear_model} can be equivalently written (without changing the notation) as 
\begin{eqnarray}
\label{new_linear_model}
\textbf{y} = \mathcal{H}\textbf{s} + \sqrt{\frac{n_{t}}{P}}\textbf{z}.
\end{eqnarray}
To the linear model in \eqref{new_linear_model}, we apply the IF linear receiver as shown in Fig. \ref{model1} to recover $\hat{{\textbf s}}$, a vector of decoded information symbols. We apply the constraint $2K \leq 2n_{r}T$ to avoid the case of under determined linear system. In order to decode an STBC using the IF receiver, the components of $\textbf{s}$ are restricted to take values from a subset of integers such that $\{ \mathcal{H}\textbf{s} ~|~ \textbf{s} \in \mathcal{S}^{2K}\}$ is a lattice code carved from the lattice $\Lambda = \left\lbrace \mathcal{H}\textbf{s} ~|~ \textbf{s} \in \mathbb{Z}^{2K} \right\rbrace.$ This is the reason for choosing the components of $\textbf{s}$ from the ring $\mathbb{Z}_{\sqrt{M}}$.

%\begin{definition}
%A $d$-dimensional {\em lattice} $\Lambda$ with a basis set $\{{\textbf g}_1,\ldots,{\textbf g}_d\}\subseteq\mathbb{R}^d$ is the set of all points of the form $\{{\textbf x}={\textbf u}{\bf G}| {\textbf u}\in \mathbb{Z}^d\}$ where ${\textbf G}$ is the {\em generator matrix} of $\Lambda$, formed by placing ${\textbf g}_m$'s as its rows.
%\end{definition}

The goal of the IF receiver is to project $\mathcal{H}$ onto a non-singular integer matrix ${\textbf A} \in \mathbb{Z}^{2K \times 2K}$ by left multiplying $\mathcal{H}$ with a receiver filtering matrix ${\textbf B} \in \mathbb{R}^{2K \times 2n_{r}T}$. After post processing by {\textbf B}, we get
\begin{equation}~\label{eq:LRmodel}
\tilde{{\textbf y}} \triangleq {\textbf B}{\textbf y} = {\textbf B}\mathcal{H}{\textbf s}+\sqrt{\frac{n_{t}}{P}}{\textbf B}{\textbf z}.
\end{equation}
The above signal model is applicable to all linear receivers including the ZF, MMSE (both cases ${\textbf A}={\textbf I}_{2K}$), and IF (where ${\textbf A}$ is invertible over $\mathcal{S}$).
For the IF receiver formulation, we write
\begin{equation}~\label{eq:IFmodel}
\tilde{{\textbf y}}= {\textbf A}{\textbf s}+({\textbf B}\mathcal{H}-{\textbf A}){\textbf s}+\sqrt{\frac{n_{t}}{P}}{\textbf B}{\textbf z},
\end{equation}
where ${\textbf A}{\textbf s}$ is the desired signal component, and the effective noise is $({\textbf B}\mathcal{H}-{\textbf A}){\textbf s}+\sqrt{\frac{n_{t}}{P}}{\textbf B}{\textbf z}$. In particular, the effective noise power along the $m$-th row (henceforth referred to as the $m$-th layer) of $\tilde{{\textbf y}}$ for $1 \leq m \leq 2K$ is defined as
\begin{equation}~\label{quntizederrplusnoise}
g({\textbf a}_m,{\textbf b}_m)\triangleq \|{\textbf b}_m \mathcal{H}-{\textbf a}_m\|^2 \bar{E} + \frac{n_{t}}{2P}\|{\textbf b}_m\|^2,
\end{equation}
where ${\textbf a}_m$ and ${\textbf b}_m$ denote the $m$-th row of ${\textbf A}$ and $\textbf{ B}$, respectively, and $\bar{E}$ is the average energy of the constellation $\mathcal{S}$. A layer based model of the IF receiver architecture is as shown in Fig. \ref{model1}. In order to reduce the effective noise power for each layer, the term $g({\textbf a}_m,{\textbf b}_m)$ has to be minimized for each $m$ by appropriately selecting the matrices ${\textbf A}$ and ${\textbf B}$. For methods to select $\textbf{A}$ and $\textbf{B}$, we refer the reader to \cite{zhan12}, \cite{Sakzad14-1}. In order to uniquely recover the information symbols, the matrix ${\textbf A}$ must be invertible over the ring $\mathcal{S}$. In this work we are only interested in the STBC design for the IF receiver and hence, we assume that the optimal values of $\textbf{A}$ and $\textbf{B}$ are readily available.

\indent We now present a procedure for decoding STBCs using the IF linear receiver. With reference to the signal model in Section II, the decoding procedure exploits the ring structure of the constellation $\mathcal{S} = \mathbb{Z}_{\sqrt{M}} = \left\{0, 1, \ldots, \sqrt{M} - 1 \right\}$ with operations mod $\sqrt{M}$. The decoding procedure is as given below:
\begin{itemize}
\item \textbf{Step 1 (Infinite lattice decoding over $\mathbb{Z}$):} Each component of $\tilde{{\textbf y}}$ is decoded to the nearest point in $\mathbb{Z}$ to get $\hat{{\textbf y}} = \lfloor \tilde{{\textbf y}}\rceil,$ where $\lfloor \cdot \rceil$ denotes the round operation.
\item \textbf{Step 2 (Modulo operation onto $\mathcal{S}$):} Perform the modulo $\sqrt{M}$ operation on the components of $\hat{{\textbf y}}$ to obtain ${\textbf r} = \left(\hat{{\textbf y}} \mbox{ mod } \sqrt{M}\right) \in \mathcal{S}^{2K}.$
\item \textbf{Step 3 (Solving system of linear equations):} Solve the system of linear equations ${\textbf r} = {\textbf A}\hat{{\textbf s}}$ over the ring $\mathbb{Z}_{\sqrt{M}}$. If $\textbf{A}$ is invertible over the ring $\mathcal{S}$, then a unique solution is guaranteed. After solving the system of linear equations, information symbols are recovered from the components of $\hat{{\textbf s}}$. 
\end{itemize}

\indent In the above decoding procedure, \textbf{Step 2} and \textbf{Step 3} are deterministic, while \textbf{Step 1} involves recovering linear functions of the information symbols amidst noise. In the next section, we obtain an upper bound on the probability of error for \textbf{Step 1}, and then derive a design criterion for full-diversity STBCs.

%%%%%%%%%%%%%%%%%%%%%%%%%%%%%%%%%%%%%%%%%%%%%%%%%%%%%%%%%%%%%%%%%%%%%%%%%%%%%%%%%%%%%%
%New Section
%%%%%%%%%%%%%%%%%%%%%%%%%%%%%%%%%%%%%%%%%%%%%%%%%%%%%%%%%%%%%%%%%%%%%%%%%%%%%%%%%%%%%

\section{Design Criterion for STBCs}
\label{sec4}

We first present an upper bound on the probability of error for \textbf{Step 1}, i.e., decoding the $m$-th layer in the infinite lattice $\mathbb{Z}$ for $1 \leq m \leq 2K$. The input to the decoder in \textbf{Step 1} is
\begin{equation*}
\tilde{{\textbf y}}_{m} = {\textbf a}_{m}{\textbf s} + ({\textbf b_{m}}\mathcal{H} - {\textbf a}_{m}){\textbf s} + \sqrt{\frac{n_{t}}{P}}\textbf{b}_{m}\textbf{z},
\end{equation*}
where $\tilde{{\textbf y}}_{m}$ denotes the $m$-th component of $\tilde{{\textbf y}}$ and $({\textbf b_{m}\mathcal{H}} - {\textbf a}_{m}){\textbf s}$ denotes the quantization noise term. %Since ${\textbf a}_{m}$ is a row vector and ${\textbf s}$ is a column vector, the operation ${\textbf a}_{m}{\textbf s}$ is well defined. 
For such a set-up, the effective noise power is given in \eqref{quntizederrplusnoise}. Note that the effective noise is not Gaussian distributed due to the quantization noise term. However, since the optimum value of $\textbf{b}_{m}$ that minimizes \eqref{quntizederrplusnoise} given $\textbf{a}_{m}$ is
\begin{equation*}
\textbf{b}_{m} = \textbf{a}_{m}\mathcal{H}^{T}\left(\frac{n_{t}}{P \bar{E}}\textbf{I}_{2n_{r}T} + \mathcal{H}\mathcal{H}^{T}\right)^{-1},
\end{equation*}
for large values of $P$, the above expression simplifies to $\textbf{b}_{m} = \textbf{a}_{m}\mathcal{H}^{-1}$, where $\mathcal{H}^{-1}$ denotes the pseudo-inverse of $\mathcal{H}$. With this, for large values of $P$, the quantization noise term vanishes and the effective noise power is approximated by
\begin{equation*}
g({\textbf a}_m,{\textbf b}_m) = \frac{n_{t}}{2P}\|{\textbf b}_m\|^2.
\end{equation*}

\noindent Since we are interested in the full-diversity property of STBCs, which is a large SNR metric, we assume large values of $P$ in the probability of error analysis. Henceforth, we denote the probability of error for decoding the $m$-th layer in the infinite lattice  $\mathbb{Z}$ by $P_{e}(m, \mathcal{H}, \mathbb{Z})$. The following Lemma provides an upper bound on $P_{e}(m, \mathcal{H}, \mathbb{Z})$.\\
\begin{lemma}\textbf{(Upper Bound on Probability of Error)}
For large values of $P$, the term $P_{e}(m, \mathcal{H}, \mathbb{Z})$ is upper bounded as
\begin{equation}
\label{P_e_bound_2}
P_{e}(m, \mathcal{H}, \mathbb{Z}) \leq \mbox{exp}\left(-cP\epsilon_{1}^{2}(\Lambda)\right),
\end{equation}
where $c$ is some constant independent of $P$ and $\epsilon_{1}^{2}(\Lambda)$ is the minimum squared Euclidean distance of the lattice $\Lambda = \left\lbrace \textbf{d}\mathcal{H}^{T} ~|~ \textbf{d} \in \mathbb{Z}^{2K} \right\rbrace.$
\end{lemma}
\begin{proof}
Since the minimum Euclidean distance of $\mathbb{Z}$ is unity, an error in \textbf{Step 1} is declared if $\sqrt{\frac{n_{t}}{P}}|\textbf{b}_{m}\textbf{z}| \geq \frac{1}{2}$. Therefore, we have
\begin{eqnarray}
P_{e}(m, \mathcal{H}, \mathbb{Z}) \triangleq \mbox{Pr}\left(\sqrt{\frac{n_{t}}{P}}|\textbf{b}_{m}\textbf{z}| \geq \frac{1}{2}\right).
\end{eqnarray}
Since $\textbf{b}_{m}\textbf{z}$ is Gaussian distributed, using the Chernoff bound, $P_{e}(m, \mathcal{H}, \mathbb{Z})$ is bounded as 
\begin{eqnarray}
P_{e}(m, \mathcal{H}, \mathbb{Z}) & \leq & \mbox{exp}\left(-{\frac{P}{4n_{t}\|{\textbf b}_m\|^2}}\right) \nonumber\\
& = & \mbox{exp}\left(-{\frac{P}{4n_{t}\|{\textbf a}_m \mathcal{H}^{-1}\|^2}}\right). \label{P_e_bound1}
\end{eqnarray}
If $\textbf{a}_{m}$ and $\textbf{b}_{m}$ are chosen appropriately as in \cite{zhan12}, then the upper bound 
\begin{equation}
\label{bound1}
\|{\textbf a}_m \mathcal{H}^{-1}\|^2 \leq \epsilon_{2K}^{2}(\Lambda^{*})
\end{equation}
holds good where $\epsilon_{2K}^{2}(\Lambda^{*})$ denotes the $2K$-th successive minimum of the dual lattice
\begin{equation*}
\Lambda^{*} = \left\lbrace \textbf{d}\mathcal{H}^{-1} ~|~ \forall \textbf{d} \in \mathbb{Z}^{2K} \right\rbrace.
\end{equation*}
Here $\mathcal{H}^{-1}$ is a generator of the dual lattice $\Lambda^{*}$ of the lattice given by $\Lambda = \left\lbrace \textbf{d}\mathcal{H}^{T} ~|~ \forall \textbf{d} \in \mathbb{Z}^{2K} \right\rbrace,$ which is generated by the rows of $\mathcal{H}^{T}$. Thus we have the relation (see Lemma $4$ in \cite{zhan12})
\begin{equation}
\label{bound2}
\epsilon_{2K}^{2}(\Lambda^{*}) \leq \frac{2K^{3} + 3K^{2}}{\epsilon_{1}^{2}(\Lambda)},
\end{equation}
where $\epsilon_{1}^{2}(\Lambda)$ is the minimum squared Euclidean distance of the lattice $\Lambda$. Using the upper bounds of \eqref{bound1} and \eqref{bound2} in \eqref{P_e_bound1}, the  probability of error for decoding the $m$-th layer is upper bounded as
\begin{equation}
\label{P_e_bound_2}
P_{e}(m, \mathcal{H}, \mathbb{Z}) \leq \mbox{exp}\left(-cP\epsilon_{1}^{2}(\Lambda)\right),
\end{equation}
where $c = \frac{1}{4n_{t}(2K^{3} + 3K^{2})}$ is a constant. This completes the proof.
\end{proof}

We now introduce a new property of linear designs to establish a relation between the upper bound in $\eqref{P_e_bound_2}$ and the structure of linear designs.

\indent An infinite STBC $\mathcal{C}_{\infty}$ generated from a linear design $\textbf{X}_{\mathcal{LD}}$ in $2K$ variables is given by
\begin{equation*}
\mathcal{C}_{\infty} \triangleq \left\lbrace \textbf{X} = \sum_{k = 1}^{2K} \textbf{D}_{k} s_{k} ~|~ s_{k} \in \mathbb{Z} ~\forall k \right\rbrace. 
\end{equation*}
As a special case, a finite STBC $\mathcal{C}$ can be obtained from $\textbf{X}_{\mathcal{LD}}$ by restricting the variables $s_{k}$ to $\mathcal{S}$ as
\begin{equation*}
\mathcal{C} \triangleq \left\lbrace \textbf{X} = \sum_{k = 1}^{2K} \textbf{D}_{k} s_{k} ~|~ s_{k} \in \mathcal{S} ~\forall k \right\rbrace. 
\end{equation*}
We let $\sigma_{min}(\textbf{X}) = \min_{1 \leq j \leq n_{t}} \sigma_{j}(\textbf{X})$ denote the minimum singular value of $\textbf{X}$. With that, the minimum singular value of $\mathcal{C}_{\infty}$ is given by
\begin{equation*}
\sigma_{min}(\mathcal{C}_{\infty}) \triangleq \inf_{\textbf{X} \in \mathcal{C}_{\infty}, \textbf{X} \neq \textbf{0}} \sigma_{min}(\textbf{X}).
\end{equation*}
Using the above definition of the minimum singular value of the infinite code $\mathcal{C}_{\infty}$, we define a special class of linear designs as follows:
%%%%%%%%%%%%%%%%%%%%%%%%%%%%%%%%%%%%%%%%%%%%%%%%%%%%%%%%%%%%%%%%%%%%%%%%%%%%%%%%%%%%
\begin{definition}\textbf{(Non-vanishing singular value property)}
A linear design $\textbf{X}_{\mathcal{LD}}$ is said to have the non-vanishing singular value (NVS) property over $\mathbb{Z}$ if the corresponding infinite STBC $\mathcal{C}_{\infty}$ satisfies $\sigma_{min}(\mathcal{C}_{\infty}) \neq 0.$
\end{definition}
%%%%%%%%%%%%%%%%%%%%%%%%%%%%%%%%%%%%%%%%%%%%%%%%%%%%%%%%%%%%%%%%%%%%%%%%%%%%%%%%%%%%
%%%%%%%%%%%%%%%%%%%%%%%%%%%%%%%%%%%%%%%%%%%%%%%%%%%%%%%%%%%%%%%%%%%%%%%%%%%%%%%%%%%%

\indent We now connect the NVS property of $\textbf{X}_{\mathcal{LD}}$ and the full-diversity property of $\mathcal{C}$ for the IF receiver in the following theorem.
%%%%%%%%%%%%%%%%%%%%%%%%%%%%%%%%%%%%%%%%%%%%%%%%%%%%%%%%%%%%%%%%%%%%%%%%%%%%%%%%%%%%
\begin{theorem}\textbf{(Full-Diversity Design Criterion)}
If the linear design $\textbf{X}_{\mathcal{LD}}$ has the NVS property, then any STBC $\mathcal{C}$ generated from $\textbf{X}_{\mathcal{LD}}$ over $\mathcal{S}$ provides full-diversity with the IF linear receiver.
\end{theorem}
\begin{proof}
We write $\epsilon_{1}^{2}(\Lambda)$ in \eqref{P_e_bound_2} as
\begin{equation*}
\epsilon_{1}^{2}(\Lambda) = \|{\textbf d} \mathcal{H}^{T}\|^2 = \|\mathcal{H}{\textbf d}^{T}\|^2
\end{equation*}
for some ${\textbf d} \in \mathbb{Z}^{2K}$. Note that the components of ${\textbf d}$ that result in $\epsilon_{1}^{2}(\Lambda)$ need not be in $\mathcal{S}$. Further, $\epsilon_{1}^{2}(\Lambda)$ can be written as
\begin{eqnarray*}
\epsilon_{1}^{2}(\Lambda) & = & \|\textbf{H}\textbf{X}\|_{F}^2 = \mbox{Trace} \left(\textbf{H}\textbf{U}\Sigma\textbf{U}^{H}\textbf{H}^{H}\right),
\end{eqnarray*}
where $\textbf{X} \in \mathcal{C}_{\infty}$, $\textbf{U}\Sigma\textbf{U}^{H}$ is a singular value decomposition of $\textbf{X}\textbf{X}^{H}$, and $\Sigma$ is the diagonal matrix comprising of the square of the singular values $\sigma_{j}(\textbf{X})$ for $1 \leq j \leq n_{t}$. By denoting $\textbf{HU} = \textbf{G}$, we write
\begin{equation*}
 \epsilon_{1}^{2}(\Lambda) = \mbox{Trace} \left(\textbf{G}\Sigma\textbf{G}^{H}\right)
 = \sum_{j = 1}^{n_{t}} \|\textbf{g}_{j}\|^2 \sigma^{2}_{j}(\textbf{X}),
\end{equation*}
where $\textbf{g}_{j}$ is the $j$-th column of $\textbf{G}$ and $\sigma_{j}(\textbf{X})$ denotes the $j$-th singular value of $\textbf{X}$, which is a function of the channel $\textbf{H}$. If the STBC has the NVS property, then for any $\textbf{d} \in \mathbb{Z}^{2K}$ we apply $\sigma^{2}_{j}(\textbf{X}) \geq \sigma^{2}_{min}(\mathcal{C}_{\infty}) ~\forall j,$ and hence,
\begin{eqnarray*}
\epsilon_{1}^{2}(\Lambda) & \geq & \sum_{j = 1}^{n_{t}} \|\textbf{g}_{j}\|^2 \sigma^{2}_{min}(\mathcal{C}_{\infty}).
\end{eqnarray*}
Using the above lower bound, we upper bound the expression in \eqref{P_e_bound_2} as
\begin{eqnarray*}
P_{e}(m, \mathcal{H}, \mathbb{Z}) \leq \mbox{exp}\left(-c P \sigma^{2}_{min}(\mathcal{C}_{\infty}) \sum_{j = 1}^{n_{t}} \|\textbf{g}_{j}\|^2\right).
\end{eqnarray*}
Since $\textbf{U}$ is a unitary matrix, the distribution of $\textbf{G}$ is same as that of $\textbf{H}$. Also, as $\sigma^{2}_{min}(\mathcal{C}_{\infty})$ is a constant and independent of $\mathcal{H}$, the random variables in the exponent are $\{ \|\textbf{g}_{j}\|^2 \}$, which are chi-square distributed with degrees of freedom $2n_{r}$. By averaging $P_{e}(m, \mathcal{H}, \mathbb{Z})$ over different realizations of $\|\textbf{g}_{j}\|^2$, we obtain
\begin{eqnarray*}
\mathbb{E}_{\mathcal{H}}[P_{e}(m, \mathcal{H}, \mathbb{Z})] \triangleq P_{e}(m, \mathbb{Z}) \leq \left(\frac{1}{1 + cP\sigma^{2}_{min}(\mathcal{C}_{\infty})}\right)^{n_{t}n_{r}}
\end{eqnarray*}
Since $P$ is dominant and $\sigma_{min}(\mathcal{C}_{\infty}) \neq 0$, $P_{e}(m, \mathbb{Z})$ is upper bounded as
\begin{eqnarray*}
P_{e}(m, \mathbb{Z}) < \left(\frac{1}{cP \sigma^{2}_{min}(\mathcal{C}_{\infty})}\right)^{n_{r} n_{t}}.
\end{eqnarray*}
With the above result on the probability of error for each layer, we now setup an upper bound on the overall probability of error for \textbf{Step 1}. We declare an error in \textbf{Step 1} if there is a decoding error in any one of the $2K$ layers. Using the union bound, the overall probability of error is bounded as
\begin{equation*}
\mbox{Pr}(\hat{{\textbf y}} \neq \textbf{As}~|~\mathcal{H}) \leq \sum_{k = 1}^{2K} P_{e}(m, \mathcal{H}, \mathbb{Z}).
\end{equation*}
After taking expectation, the average probability of error for decoding \textbf{Step 1} is
\begin{eqnarray}
\mbox{Pr}(\hat{{\textbf y}} \neq \textbf{As}) & \triangleq & \mathbb{E}_{\mathcal{H}}[\mbox{Pr}(\hat{{\textbf y}} \neq \textbf{As}~|~\mathcal{H})] \nonumber \\
& \leq & \sum_{k = 1}^{2K} \mathbb{E}_{\mathcal{H}}[P_{e}(m, \mathcal{H}, \mathbb{Z})] \nonumber \\
& = & \sum_{k = 1}^{2K} P_{e}(m, \mathbb{Z}) \nonumber \\
& \leq & \frac{c'}{P^{n_{r} n_{t}}} \label{over_all_bound},
\end{eqnarray}
where $c' = \frac{2K}{(c\sigma^{2}_{min}(\mathcal{C}_{\infty}))^{n_{t}n_{r}}}$. Notice that the upper bound in \eqref{over_all_bound} is only a function of $\mathcal{C}_{\infty}$, and it is independent of the constellation $\mathcal{S}$. This shows that any STBC carved from a linear design with the NVS property provides diversity of $n_{t}n_{r}$ independent of the size of $\mathcal{S}$.
\end{proof}

\section{Simulation Results on STBCs for IF Receivers}
Through simulation results, we show that a linear design with the NVS property provides full-diversity for the IF receiver. 
%We pick two designs of identical symbol-rate that satisfy the NVS property and show that the one with larger $\sigma_{min}(\mathcal{C}_{\infty})$ performs better than the other. 
We use the Alamouti design given by 
\begin{equation}
\label{example_code_with_nvs}
\textbf{X}_{A} = \left[\begin{array}{cc}
x_{1} & x_{2}\\
-x^{*}_{2} & x^{*}_{1}\\
\end{array}\right],
\end{equation}
to showcase the results. Using the structure of the above design, it can be shown that $\sigma^{2}_{min}(\textbf{X}_{A}) = |x_{1}|^2 + |x_{2}|^{2}.$ From the above expression, it is straightforward to observe that $\sigma_{min}(\mathcal{C}_{\infty}) = 1$ for $\textbf{X}_{A}$. In Fig. \ref{ber1}, we present the bit error rate (BER) of the Alamouti code for the $2 \times 1$ MIMO channel when decoded with (\emph{i}) the IF receiver, and (\emph{ii}) the ML decoder. The plots confirm that the Alamouti code provides full-diversity with the IF receiver. For the simulation results, the method proposed in \cite{Sakzad14-1} is used throughout the paper to compute the $\textbf{A}$ and $\textbf{B}$ matrices for the IF receiver. It is well known that Alamouti code is ML decodable with lower computational complexity than the IF receiver. Despite its increased complexity, we have used the IF receiver for Alamouti code only to demonstrate that linear designs with the NVS property provide full diversity for the IF linear receiver.

\begin{figure}
\includegraphics[scale = 0.6]{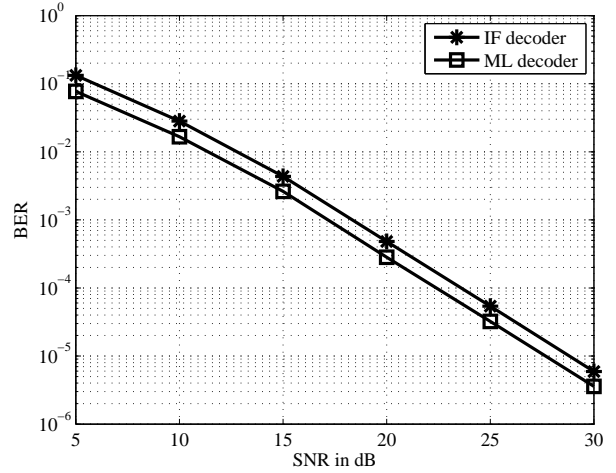}
\vspace{-0.5cm}
\caption{BER comparison of Alamouti code with IF linear receiver and the ML decoder.}
\label{ber1}
\end{figure}

%\begin{figure}
%\includegraphics[width=3.5in]{Example_code}
%\caption{BER comparison of the example code in \eqref{example_code} with IF linear receiver and the ML decoder.}
%\label{ber2}
%\end{figure}

\section{Diversity Results of IF receiver for the V-BLAST Encoding Scheme}
\label{sec5}

\begin{figure}
\includegraphics[scale = 0.36]{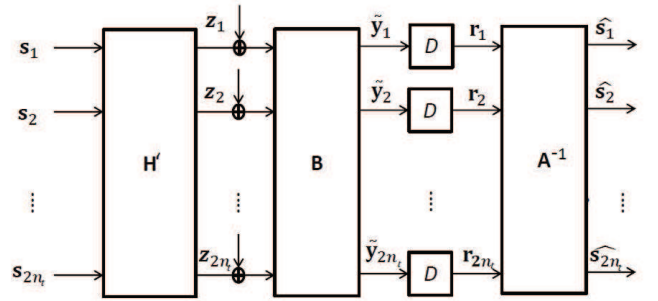}
\vspace{-0.6cm}
\caption{IF receiver for the V-BLAST scheme where $\textbf{H}'$ is given in \eqref{H_dash}.}
\label{IF_model_no_stbc}
\end{figure}

In the seminal paper \cite{zhan12} on the IF receiver, a V-BLAST scheme is employed  at the transmitter where the information symbols transmitted across different antennas are independent. In particular, the authors of  \cite{zhan12} were interested in characterizing the diversity multiplexing trade-off (DMT) of the MIMO system with IF receiver. In order to derive the DMT, the system analysis in \cite{zhan12} is based on the outage probability. Hitherto, we are not aware of any results that present the diversity results of the IF receiver based on the approach of \cite{TSC}, \cite{ShX}.

\indent In this section we present an error probability analysis for the receiver architecture introduced in \cite{zhan12} shown in Fig. \ref{IF_model_no_stbc}, and show that the IF linear receivers provide full receive diversity. This analysis is a special case of the analysis discussed in Section \ref{sec4}. We consider the uncoded case i.e, $T = 1$ for the layered architecture in our analysis. With reference to the setting discussed in Section \ref{sec2}, the uncoded layered architecture corresponds to the use of a linear design
\begin{equation*}
\textbf{X}_{\mathcal{LD}}(s_{1}, \ldots, s_{2n_{t}}) = \sum_{k = 1}^{2n_{t}} \textbf{D}_{k} s_{k},
\end{equation*}
where $\textbf{D}_{k} \in \{ \textbf{e}_{1}, \imath\textbf{e}_{1}, \textbf{e}_{2}, \imath\textbf{e}_{2}, \ldots, \textbf{e}_{n_{t}}, \imath\textbf{e}_{n_{t}} \}$ such that $\{ \textbf{e}_{k} \}_{k = 1}^{n_{t}}$ denotes the standard basis set in $\mathbb{R}^{n_{t}}$. Similar to Section \ref{sec2}, the received matrix $\textbf{Y} \in \mathbb{C}^{n_{r} \times 1}$ in \eqref{signal_model} can be vectorized to obtain a noisy linear model from $\mathbb{R}^{2n_{t}}$ to $\mathbb{R}^{2n_{r}}$ as
\begin{eqnarray}
\label{linear_model_new}
\textbf{y} = \sqrt{\frac{P}{n_{t}}}\mathbf{\mathcal{H}}\textbf{s} + \textbf{z},
\end{eqnarray}
where $\mathcal{H}$ becomes $\textbf{H}'$ after replacing $T = 1$ and $\textbf{R} =  \textbf{I}_{2n_{t}}$ in \eqref{H_dash1}. 
%After suitable scaling, \eqref{linear_model_new} can be equivalently written (without changing the notation) as 
%\begin{eqnarray}
%\label{new_linear_model_new}
%\textbf{y} = \mathcal{H}\textbf{s} + \sqrt{\frac{n_{t}}{P}}\textbf{z}.
%\end{eqnarray}
After applying the decoding procedure in Section \ref{sec2} and following the analysis in Section \ref{sec4}, we upper bound the probability of decoding error for the $m$-th layer in \textbf{Step 1} as
\begin{equation}
\label{P_e_bound_3}
P_{e}(m, \mathcal{H}, \mathbb{Z}) \leq \mbox{exp}\left(-cP\epsilon_{1}^{2}(\Lambda)\right),
\end{equation}
where $1 \leq m \leq 2n_{t}$, $c = \frac{1}{4n_{t}(2n_{t}^{3} + 3n_{t}^{2})}$ is some constant and $\epsilon_{1}^{2}(\Lambda)$ is the minimum squared Euclidean distance of the lattice $\Lambda = \left\lbrace \textbf{d}\mathcal{H}^{T} ~|~ \forall \textbf{d} \in \mathbb{Z}^{2n_{t}} \right\rbrace.$

Unlike Section \ref{sec2} where $n_{t} \times T$ dimensional matrices are codewords, in this case, $n_{t} \times 1$ dimensional complex vectors are codewords. We now discuss the singular value properties of these vectors to assist the proof for diversity order of the IF receiver. For any $\textbf{x} \in \mathbb{Z}[\imath]^{n_{t}}$, let $\sigma_{1}(\textbf{x})$ denote the non-zero singular value of $\textbf{x}$. We define the minimum non-zero singular value of $\mathcal{C}_{\infty} = \mathbb{Z}[\imath]^{n_{t}}$ as $\sigma_{min}(\mathcal{C}_{\infty}) = \inf_{\textbf{x} \in \mathcal{C}_{\infty}, \textbf{x} \neq \textbf{0}} \sigma_{1}(\textbf{x}).$ We have the following result on $\sigma_{min}(\mathcal{C}_{\infty})$.

\begin{lemma}
\label{lemma2}
For $\mathcal{C}_{\infty} = \mathbb{Z}[\imath]^{n_{t}}$, we have $\sigma_{min}(\mathcal{C}_{\infty}) \geq 1$.
\end{lemma}
\begin{proof}
For any non-zero $\textbf{x} \in \mathbb{Z}[\imath]^{n_{t}}$, the matrix $\textbf{x}\textbf{x}^{H}$ has rank one, and hence, there is only one non-zero singular value of $\textbf{x}$. This means we have $\sigma_{j}(\textbf{x}) = 0 ~\mbox{ for } j \neq 1$ and $\sigma_{1}(\textbf{x}) \geq 0$ for any non-zero $\textbf{x} \in \mathbb{Z}[\imath]^{n_{t}}$. Further, since the Trace property is preserved among \emph{similar} matrices, we have $\sum_{j = 1}^{n_{t}} \sigma_{j}(\textbf{x}) = \sigma_{1}(\textbf{x}) = \mbox{Trace}(\textbf{x}\textbf{x}^{H}) = \|\textbf{x}\|^2$. Thus $\sigma_{1}(\textbf{x}) \geq 1$ for any non-zero $\textbf{x} \in \mathbb{Z}[\imath]^{n_{t}}$. This completes the proof.
\end{proof}

Using the above Lemma, the diversity order of the IF receiver is derived in the following theorem. 

\begin{theorem}
For the V-BLAST scheme, IF linear receiver provides full-receive diversity.
\end{theorem}
\begin{proof}
We can write $\epsilon_{1}^{2}(\Lambda)$ in \eqref{P_e_bound_3} as
\begin{equation*}
\epsilon_{1}^{2}(\Lambda) = \|{\textbf d} \mathcal{H}^{T}\|^2 = \|\mathcal{H}{\textbf d}^{T}\|^2
\end{equation*}
for some ${\textbf d} \in \mathbb{Z}^{2n_{t}}$. Note that the components of ${\textbf d}$ that result in $\epsilon_{1}^{2}(\Lambda)$ need not be in $\mathcal{S}$. Further, we can write $\epsilon_{1}^{2}(\Lambda)$ as
\begin{eqnarray*}
\epsilon_{1}^{2}(\Lambda) & = & \|\textbf{H}\textbf{x}\|_{F}^2 = \mbox{Trace} \left(\textbf{H}\textbf{U}\Sigma\textbf{U}^{H}\textbf{H}^{H}\right),
\end{eqnarray*}
where $\textbf{x} \in \mathcal{C}_{\infty}$, $\textbf{U}\Sigma\textbf{U}^{H}$ is a singular value decomposition of $\textbf{x}\textbf{x}^{H}$, and $\Sigma$ is the diagonal matrix containing the square of the singular values $\sigma_{j}(\textbf{x})$ for $1 \leq j \leq n_{t}$. By denoting $\textbf{HU} = \textbf{G}$, we have
\begin{equation*}
 \epsilon_{1}^{2}(\Lambda) = \mbox{Trace} \left(\textbf{G}\Sigma\textbf{G}^{H}\right)
 = \sum_{j = 1}^{n_{t}} \|\textbf{g}_{j}\|^2 \sigma^{2}_{j}(\textbf{x}),
\end{equation*}
where $\sigma_{j}(\textbf{x})$ denotes the $j$-th singular value of $\textbf{x}$, which is a function of the channel $\textbf{H}$. Since $\textbf{x}$ is a rank one vector, we have $\sigma_{j}(\textbf{x}) = 0 ~\mbox{ for } j \neq 1$ and from Lemma \ref{lemma2}, we have $\sigma_{1}(\textbf{x}) \geq 1$. Hence $\epsilon_{1}^{2}(\Lambda) \geq \|\textbf{g}_{1}\|^2$ where $\textbf{g}_{1}$ denotes the $1$-st column of $\textbf{G}$. Using the above lower bound, we upper bound the expression in \eqref{P_e_bound_3} as
\begin{eqnarray*}
P_{e}(m, \mathcal{H}, \mathbb{Z}) \leq \mbox{exp}\left(-c P \|\textbf{g}_{1}\|^2\right).
\end{eqnarray*}
By averaging $P_{e}(m, \mathcal{H}, \mathbb{Z})$ over different realizations of $\|\textbf{g}_{1}\|^2$, we obtain
\begin{eqnarray*}
\mathbb{E}_{\mathcal{H}}[P_{e}(m, \mathcal{H}, \mathbb{Z})] \triangleq P_{e}(m, \mathbb{Z}) \leq \left(\frac{1}{1 + cP}\right)^{n_{r}}
\end{eqnarray*}
Thus the V-BLAST scheme in \cite{zhan12} provides full-receive diversity for IF linear receivers.   
\end{proof}

\section{Directions for Future Work}
\label{sec7}

%\begin{small}
%\begin{table}
%\begin{center}
%\caption{Comparison of STBCs for various receivers}
%\begin{tabular}{|c|c|c|c|} \hline
%Approach & Decoding & Symbol- & Equalization\\
%& Complexity & Rate & Complexity\\
%\hline
%%\multirow{5}{*}{Integer-Forcing Linear receiver}
%{\textbf ML} & high & $\leq \mbox{min}(n_{t}, n_{r})$ & low\\
%{\textbf ZF \& MMSE} & low & $\leq 1$ & medium\\
%%{\textbf ZF \& MMSE} & low & $n_{t} - n_{r} + 1$ & $\mbox{min}(n_{t}, n_{r})$\\
%{\textbf IF} & low & $\leq \mbox{min}(n_{t}, n_{r})$ & high\\
%\hline
%\end{tabular}
%\end{center}
%\label{updated_complexity_table}
%\end{table}
%\end{small}

%We have presented a decoder analysis for the IF receiver in order to obtain a design criterion for full-diversity STBCs. We have shown that STBCs that satisfy the NVS criterion provide full-diversity for the IF linear receiver. 
%We have also shown that STBCs with the NVD criterion can be used for the IF receiver. In summary, among the class of linear receivers, IF receivers admit STBCs with larger spectral efficiency than that of the MMSE and ZF receivers. 
%To conclude, we list down various parameters of the linear receivers in Table \ref{updated_complexity_table} which shows that the reduction in the decoding complexity for IF receivers comes at the cost of increased complexity in equalization in comparison with the ML decoder.
An interesting direction for future work is to verify whether any known classes of high-rate STBCs (such as Golden code) provide full diversity with IF linear receivers. In such a case, the trade-off in the error performance can be studied between the low-complexity IF receivers and the high complexity ML decoder. Another direction of future work is to construct new STBC designs with large value of $\sigma_{min}(\mathcal{C}_{\infty})$ so that they perform well with IF receivers.
%%%%%%%%%%%%%%%%%%%%%%%%%%%%%%%%%%%%%%%%%%%%%%%%%%%%%%%%%%%
%%%%%%%%%%%%%%%%%%%%%%BIBLIOGRAPHY%%%%%%%%%%%%%%%%%%%
\section*{Acknowledgements}

The authors would like to thank Roope Vehkalahti and Lakshmi Prasad Natarajan for their comments on the NVS property of STBCs. J. Harshan is supported by the Human-Centered Cyber-physical Systems Programme at the Advanced Digital Sciences Center from Singapore's Agency for Science, Technology and Research (A*STAR). Amin Sakzad and Emanuele Viterbo are supported by the Australian Research Council (ARC) under Discovery grants ARC DP150100285 and ARC DP160101077, respectively

\end{document}